\documentclass[11pt,onecolumn,draftclsnofoot]{IEEEtran} 

\usepackage{graphicx, cite, verbatim, url}
\usepackage{stackrel}

\newtheorem{theorem}{Theorem}

\newtheorem{subsec:coding}{subsec:coding}
\newtheorem{fact}{Fact}

\newtheorem{lemma}{Lemma}

\newcommand{\ls}[1]  
   {\dimen0=\fontdimen6\the=#1\dimen0
    \advance\lineskip.5\fontdimen5\the\lineskip-\dimen0
    \lineskiplimit=.9\lineskip
    \baselineskip=\lineskip
    \advance\baselineskip\dimen0
    \normallineskip\lineskip
    \normallineskiplimit\lineskiplimit
    \normalbaselineskip\baselineskip
    \ignorespaces
   }
   
\begin{document}

\title{Downlink Power Allocation for Stored Variable-Bit-Rate Videos}

\author{Yingsong~Huang, \ Shiwen Mao, \ and Yihan Li
\thanks{Y. Huang, S. Mao, and Y. Li are with the Department of Electrical and Computer Engineering, Auburn University, Auburn, AL 36849-5201. Email: dzh0003@tigermail.auburn.edu, smao@ieee.org, yli@auburn.edu.}
\thanks{Part of this work was presented at ICST QShine 2010~\cite{Huang10}.} 
}

\maketitle

\begin{abstract}        
In this paper, we study the problem of power allocation for streaming multiple variable-bit-rate (VBR) videos in the downlink of a cellular network. We consider a deterministic model for VBR video traffic and finite playout buffer at the mobile users.  The objective is to derive the optimal downlink power allocation for the VBR video sessions, such that the video data can be delivered in a timely fashion without causing playout buffer overflow and underflow. The formulated problem is a nonlinear nonconvex optimization problem. We analyze the convexity conditions for the formulated problem and propose a two-step greedy approach to solve the problem.  We also develop a distributed algorithm based on the dual decomposition technique. The performance of the proposed algorithms are validated with simulations using VBR video traces under realistic scenarios. 
\end{abstract}

\pagestyle{plain}\thispagestyle{plain}

\begin{keywords} 
Convex optimization; distributed algorithm, downlink power control, video streaming, variable bit rate video.
\end{keywords}

\section{Introduction}


According to a recent study by Cisco, data traffic over wireless networks is expected to increase by a factor of 66 times by 2013.  Much of the increase in future wireless data traffic will be video related, as driven by the compelling need for ubiquitous access to multimedia content for mobile users.  Such drastic increase in video traffic will significantly stress the capacity of existing and future wireless networks. While new wireless network architectures and technologies are being developed to meet this ``grand challenge''~\cite{Zhao09, Su08, Alay09, Mao05}, it is also important to revisit existing wireless networks, to maximize their potential in carrying real-time video data. 

Quality of service guarantee in wireless networks is a challenging problem that has attracted tremendous efforts~\cite{Su08JNL, Tang07JSAC, Tang08TW, Zhang06CM, Tang07TW, Zhang02TON}. 
In this paper, we consider the problem of streaming multiple videos in the downlink of a cellular network. The system is interference limited: the capacity of a specific mobile user depends on the Signal to Interference-plus-Noise Ratio (SINR) at the user, which is a function of the power allocation for all the mobile users. Therefore, effective downlink power control is necessary for such a wireless video system to minimize the intra-cell interference for concurrent video sessions. 

We consider the challenging problem of streaming concurrent variable-bit-rate (VBR) videos in the cellular network. This is motivated by the fact that VBR video offers stable and superior quality over constant bit rate (CBR) videos.  Furthermore, many stored video content are VBR.  It is important to support such stored VBR videos over existing wireless networks without the need for transcoding. 
The main challenge in supporting VBR video stems from its high rate variability and complex autocorrelation structure, making it hard for network control and may cause frequent playout buffer underflow or overflow.  In this paper, we adopt a deterministic traffic model for stored VBR video, which jointly considers frame size, frame rate, and playout buffer size~\cite{Liew97, Salehi98, Feng00}. Unlike 
prior work, we exploit effective downlink power control to adjust the downlink capacities based on prior knowledge of frame sizes and palyout buffer occupancies. 


Specifically, we present a downlink power control framework for streaming multiple VBR videos in a cellular network.  With the deterministic VBR video traffic model, we formulate an optimization problem that jointly considers donwlink power control, intra-cell interference, VBR video traffic characteristics, playout buffer underflow and overflow constraints, and base station (BS) peak power constraint.  The objective is to maximize the total throughput, which can achieve high playout buffer utilization.  As a result, playout buffer underflow or overflow events can be minimized.  We analyze the convex/concave regions of the formulated problem and develop a two-step downlink power allocation algorithm for solving the problem.  We also develop a distributed algorithm based on the dual decomposition technique from convex optimization, in order to reduce the control and computation overhead at the BS.  We evaluate the performance of the proposed distributed algorithm with simulations using VBR video traces. Our simulation results verify the accuracy of the analysis and demonstrate the efficacy of the proposed algorithms.

The remainder of this paper is organized as follows. The system model is presented in Section~\ref{sec:sys}. We develop a two-step algorithm to solve the power allocation problem in Section~\ref{sec:PowerFrm}, and a distributed 
algorithm based on 
dual decomposition in Section~\ref{sec:DistPower}. 
Simulation results are presented in Section~\ref{sec:simulation} and related work are discussed in Section~\ref{sec:related}. Section~\ref{sec:conclusion} concludes this paper.

\section{Network and Video System Model} \label{sec:sys}

We consider the downlink of a cellular network. 
In the cell, a BS streams multiple VBR videos simultaneously to mobile users in the cell, which share the downlink bandwidth.  We assume the last-hop wireless link is the bottleneck, while the wired segment of a session path is reliable with sufficient bandwidth. 
Thus the corresponding video data is always available at the BS before the scheduled transmission time.  

VBR video traffic exhibits both strong asymptotic self-similarity and short-range correlation. 
A stochastic model capturing the complex auto-correlation structure 
often requires a large number of parameters, and is thus hard to be incorporated for 
scheduling real-time video data.  To this end, we adopt a {\em deterministic model} 
that considers frame sizes and playout buffers~\cite{Salehi98}. 
Let $D_n(t)$ denote the {\em cumulative consumption curve} 
of the $n$-{th} mobile user, representing the total amount of bits consumed by the decoder at time $t$.  The cumulative consumption curve is determined by video characteristics such as frame sizes and frame rates.
%
Assume mobile user $n$ has a playout buffer of size $b_n$ bits and its video has
$T_n$ frames. 
We can derive a {\em cumulative overflow curve} for the user as 
  \begin{equation} \label{eq:b}
    B_n(t) = \min\{D_n(t-1)+b_n, D_n(T_n)\}, \; 0 \leq t \leq T_n.
  \end{equation}
$B_n(t)$ is the the maximum number of bits that can be received at time $t$ 
without overflowing user $n$'s playout buffer.  
Finally we define {\em cumulative transmission curve} $X_n(t)$ as the cumulative amount of transmitted bits to user $n$ at time $t$.  To simplify notation, we assume the video sessions have identical frame rate and the frame intervals are synchronized.  
A time slot $t$ is equal to the $t$-th frame interval, denoted as $\tau$.

The three curves for user $n$ are illustrated in Fig.~\ref{fig:vbrtrans}.  A feasible transmission schedule will produce a cumulative transmission curve $X_n(t)$ that lies within $D_n(t)$ and $B_n(t)$, i.e., causing neither underflow nor overflow at the playout buffer.  In practice, $D_n(t)$'s are known for stored videos and can be delivered to the BS during session setup phase, and $B_n(t)$'s can be derived as in (\ref{eq:b}).  

	\begin{figure} [t] 
	\center{\includegraphics[width=3.6in]{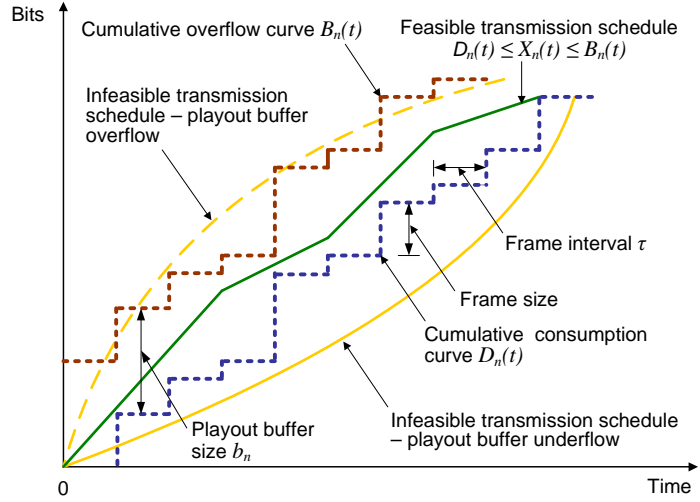}}
	\caption{Transmission schedules for VBR video session $n$.} 
	\label{fig:vbrtrans}
	\end{figure}

We consider $N$ subscribers in the cell and let $\mathcal{U}$ denote the set of users. 
In each time slot $t$, the BS transmits to each user $n$ with power $P_n(t)$ and the {\em power allocation} is $\vec{P}(t) = \left[ P_1(t), 
\cdots, P_n(t) \right]^T$. We also consider a {\em maximum transmit power} constraint $\bar{P}$, i.e., $\sum_{n \in \mathcal{U}} P_n(t) \leq \bar{P}$, for all $t$. 
When the power allocation $\vec{P}(t)$ is determined, 
the SINR at user $n$ can be written as~\cite{Lee05, Lee09}
\begin{equation} \label{eq:sinr}
  \gamma_n(\vec{P}(t)) = \frac{L_n G_nP_n(t)}{\beta\sum_{k \neq n}{G_n P_k(t)} + \eta_{n}},
\end{equation}
where $P_n$ is the power allocated to user $n$, $G_{n}$ is the path gain between the BS and user $n$,  $\eta_n$ is the noise power at user $n$, $L_n$ is a constant for user $n$ (e.g., processing gain),  and $\beta$ denotes the orthogonality factor, with $0 \leq \beta \leq 1$.  In this paper, we consider the case $\beta = 1$, where the SINR of a user not only depends on its own power allocation but also the power allocations of other users.

We assume slow-fading channels such that the path gains do not change within each time slot~\cite{Lee05}.  The downlink capacity $C_n(t)$ depends on the SINR at user $n$, the channel bandwidth $B_w$, and the transceiver design, such as modulation and channel coding.  Without loss of generality, we use the upper bound as predicted by Shannon's Theorem:
\begin{equation} \label{eq:cm}
  C_n(\vec{P}(t)) = B_w \log \left( 1 + \gamma_n(\vec{P}(t)) \right).
\end{equation} 

In time slot $t$, 
$C_n(t)\tau$ bits of video data will be delivered to user $n$. 
The cumulative transmission curve $X_n(t)$ is
\begin{equation} \label{eq:xt}
  X_n(0) = 0; \;\; X_n(t) = X_n(t-1) + C_n(t) \tau.
\end{equation}
For a feasible power allocation, the cumulative transmission curves should satisfy 
\begin{equation} \label{eq:bxd}
  D_n(t) \leq X_n(t) \leq B_n(t), \mbox{ for all } n, t,
\end{equation}
i.e., without causing playout buffer underflow or overflow.  

From (\ref{eq:cm})$\sim$({\ref{eq:bxd}), the lower and upper limit on the feasible SINR at user $n$ can be derived as
\begin{equation} \label{eq:sinr-highlow}
  \left\{ 
  \begin{array}{l}
    \gamma_{n}^{min}(t) = \max \left\{ \exp\left\{\frac{\max\left\{0, D_n(t)-X_n(t-1)\right\}}{B_w \tau} \right\}, \gamma_{n}^{th} \right\}  \\   
    \gamma_{n}^{max}(t) = \exp \left\{ \frac{B_n(t)-X_n(t-1)}{B_w \tau} \right\}, 
          \end{array} \right.  
\end{equation}
where $\gamma_{n}^{th}$ is the minimum SINR requirement imposed by the transceiver design. 
$\gamma_{n}^{min}(t)$ is the SINR that the just empties the buffer at the end of time slot $t$, without causing underflow; $\gamma_{n}^{max}(t)$ is the SINR that just fills up the buffer at the end of time slot $t$, without causing overflow. 

Generally, feasible power allocation $\vec{P}(t)$ is not unique for a given set of VBR video sessions. Among the set of feasible solutions, a schedule that transmits more data is more desirable since it provides 
more flexibility for optimizing future power allocations.  We formulate the problem of optimal downlink power control for VBR videos, termed Problem \textbf{A}, as
\begin{eqnarray} 
  \mbox{\bf (A) \bf maximize} && 
  \mbox{$\sum_{n \in \mathcal{U}}$} \log(1 + \gamma_n(t)) 
         \label{eq:optimal} \\
  \mbox{\bf subject to:} && \nonumber \\
  && \gamma_n(t) = \frac{L_n G_{n} P_n(t)}{\sum_{k \neq n}{G_{n}P_k(t)} + \eta_{n}}, \; \mbox{for all} \; n 
         \label{eq:vbrct1} \\
    && 
    \gamma_{n}^{min}(t) \leq \gamma_n(t) \leq \gamma_{n}^{max}(t),
       \; \mbox{for all} \; n \label{eq:vbrct2} \\
    && 
    \mbox{$\sum_{n \in \mathcal{U}}$} {P_n} \leq \bar{P}. \label{eq:vbrct3} 
\end{eqnarray}

In Problem \textbf{A}, 
the objective is to achieve the maximum buffer uitilization at the users, under playout buffer underflow and overflow constraints and BS maximum transmit power constraints.  This is a nonlinear nonconvex problem, to which traditional convex optimization techniques cannot directly apply. 
Due to the large variability of VBR traffic, the SINRs may assume values ranging from very low to very high, to avoid playout buffer underflow and overflow.  Thus the existing high SINR approximation~\cite{Chiang05} and low SINR approximation~\cite{Gjendemsj08} techniques cannot be directly applied. 

\section{Two-Step Downlink Power Allocation} \label{sec:PowerFrm}

In Problem \textbf{A}, we consider an interference-limited system, where the capacity of downlink $n$ depends on the power allocations for all the users.  
In the following, we first derive conditions for the optimal solution, and then present a two-step power allocation algorithm for solving Problem {\bf A}. 

\begin{lemma} \label{lm1}
If there exists a feasible power allocation $\vec{P}(t)$ that achieves  $\gamma_n^{max}(t)$ for all $n$, the solution is optimal. 
\end{lemma}
%
\begin{proof}
If the feasible power allocation $\vec{P}(t)$ achieves $\gamma_n^{max}(t)$ for all $n$, then all the user buffers are full at the end of the time slot, according to (\ref{eq:sinr-highlow}). The objective value (\ref{eq:optimal}) cannot be further improved without causing buffer overflow. Thus the solution is optimal. 
\end{proof}

\begin{lemma} \label{lm2}
If the upper limit $\gamma_n^{max}(t)$ cannot be achieved for every user $n$, then the optimal power allocation $\vec{P}(t)$ satisfies $\sum_{n \in \mathcal{U}}{P_n(t)} = \bar{P}$. 
\end{lemma}
\begin{proof}
Consider a feasible power allocation $\vec{P}'(t) = [ P'_1(t), P'_2(t),\cdots, P'_N(t) ]^T$ and $\sum_{n \in \mathcal{U}}{P'_n(t)} < \bar{P}$.  We can construct another feasible power allocation $\vec{P}''(t) = [P''_1(t), P''_2(t), \cdots, P''_N(t)]^T$, such that $P''_n(t) = \kappa \cdot P'_n(t)$, for all $n$, and $\kappa \cdot \sum_{n \in \mathcal{U}}{P'_n(t)} = \sum_{n \in \mathcal{U}}{P''_n(t)} \leq \bar{P}$, where $\kappa > 1$. For the SINR at user $n$, we have
\begin{eqnarray}
\gamma_n(\vec{P}''(t)) 
&=& \frac{L_n G_n P''_n(t)}{\sum_{k \neq n}{G_n P''_k(t)} 
+ 
\eta_n} \nonumber \\
&=& \frac{\kappa L_n G_n  P'_n(t)}{\sum_{k \neq n}{\kappa G_n} P'_k(t) 
+ 
\eta_n} \nonumber \\ 
&>& \frac{{\kappa}L_n G_n P'_n(t)}{\sum_{k \neq n} \kappa G_n P'_k(t) + \kappa \eta_n} \nonumber \\
&=& \gamma_n(\vec{P}'(t)). \nonumber
\end{eqnarray}
It follows that $\sum_{n \in \mathcal{U}}\log(1 + \gamma_n(\vec{P}''(t))) > \sum_{n \in \mathcal{U}}\log(1 + \gamma_n(\vec{P}'(t)))$, since 
$\log(1+x)$ is an increasing function of $x$. 

Choosing $\kappa = \bar{P} / \sum_{n \in \mathcal{U}}{P'_n(t)}$, we can construct a feasible solution $\vec{P}'''(t) = \kappa \cdot \vec{P}'(t)$, such that $\sum_{n \in \mathcal{U}}{P'''_n(t)} = \bar{P}$. Then we have $\gamma_n(\vec{P}'''(t)) > \gamma_n(\vec{P}'(t))$ and $\sum_{n \in \mathcal{U}}\log(1 + \gamma_n(\vec{P}'''(t))) > \sum_{n \in \mathcal{U}}\log(1 + \gamma_n(\vec{P}'(t)))$. That is, any feasible solution with $\sum_{n \in \mathcal{U}}{P'_n(t)} < \bar{P}$ will be dominated by feasible solutions with $\sum_{n \in \mathcal{U}}{P'''_n(t)} = \bar{P}$. 
We conclude that the optimal solution $\vec{P}(t)$ must satisfy $\sum_{n \in \mathcal{U}}{P_n(t)} = \bar{P}$.
\end{proof} 

We have the following result for the optimal solution of Problem {\bf A},
which directly follows Lemmas~\ref{lm1} and~\ref{lm2}. 
 
\begin{theorem} \label{th1}
A solution to Problem {\bf A} is optimal if 
(i) it achieves the maximum SINR $\gamma_n^{max}(t)$ for all $n$; or 
(ii) its total transmit power is $\bar{P}$. 
\end{theorem}
%

Theorem~\ref{th1} implies that we can examine the SINR (or buffer) constraints and the peak power constraint separately. In the rest of this section, we present a two-step power allocation algorithm for solving Problem {\bf A}.  We first examine Problem \textbf{A} under condition (i) in Theorem~\ref{th1}, to 
obtain Problem \textbf{B} as
\begin{eqnarray} 
  \mbox{\bf (B)} && 
  \gamma_n^{max}(t) = \frac{L_n G_nP_n(t)}{\sum_{k \neq n}{G_n P_k(t)} + \eta_{n}}, \; \mbox{for all} \; n, \label{eq:opt1} \\
  && \mbox{\bf subject to:} \nonumber \\
  && \;\;\;\;\;\;\;\;\;\;\;\; \mbox{$\sum_{n \in \mathcal{U}}$} {P_n} \leq \bar{P}.
\end{eqnarray}

In Problem \textbf{B}, (\ref{eq:opt1}) is a system of linear equations of power allocation $\vec{P}(t)$. Rearranging the terms, we can rewrite (\ref{eq:opt1}) in the matrix form as:
\begin{eqnarray} \label{eq:matrixform}
   \left( \mathbf{I} - \mathbf{F} \right) \vec{P}(t) =  \vec{u}, \;\; \mbox{for } \vec{P}(t) \succ \vec{0},
\end{eqnarray}
where $\mathbf{I}$ is the {\em identity matrix}, $\mathbf{F}$ is a $N \times N$ matrix with 
\begin{equation} \label{eq:F}
F_{nm} = \left\{ \begin{array}{ll}
				0, & \mbox{if } n = m \\
				\gamma_n^{max} / L_n, & \mbox{otherwise},
          \end{array}
    \right.
\end{equation}
and $\vec{u} = [\eta_1\gamma_1^{max}/L_n G_1, \eta_2\gamma_2^{max}/L_n G_2, \cdots, \eta_{N}\gamma_N^{max}/L_n G_N]^T$.  

Since all the variables are nonnegative, $\mathbf{F}$ is a non-negative matrix. According to the Perron-Frobenius Theorem, we have the following equivalent statements~\cite{Mitra93}:
\begin{fact} \label{fact1}
The following statements are equivalent: (i) there exits a feasible power allocation satisfying (\ref{eq:matrixform}); (ii) the spectrum radius of $\mathbf{F}$ is less than 1; (iii) the reciprocal matrix $(\mathbf{I} - \mathbf{F})^{-1} = \sum_{k=0}^{\infty} \left( \mathbf{F} \right)^k$ exists and is component-wise positive. 
\end{fact}


\begin{table} [t] 
\begin{center}
\caption{Two-Step Power Allocation Algorithm: Step I}
\label{tab:heuristic1}
\begin{tabular}{ll}
	\hline
	1  & BS obtains $b_n$, $D_n$, and $B_n$, and computes 
	     $\gamma_n^{max}$ for all user $n$; \\  
	2  & BS tests the existence of feasible solutions using 
	     (\ref{eq:matrixform}); \\
  3  & IF (\ref{eq:matrixform}) is solvable, compute its solution 
       $\vec{P}(t)$; \\
     & ELSE, go to Step II of the algorithm, as given in
       Table~\ref{tab:heuristic2}; \\
	4  & IF $\sum_{n \in \mathcal{U}}{P_n(t)} \leq \bar{P}$, stop with the optimal solution $\vec{P}(t)$; \\
	   & ELSE go to Step II of the algorithm, as given in
       Table~\ref{tab:heuristic2}; \\
	\hline
\end{tabular}
\end{center}
\end{table}

Based on Theorem~\ref{th1} and Fact~\ref{fact1}, we derive the {\em first step} of the two-step power allocation algorithm,  
as given in Table~\ref{tab:heuristic1}.  If Problem {\bf B} is solvable, the Step I algorithm in Table~\ref{tab:heuristic1} produces the optimal solution for Problem {\bf A} according to Theorem~\ref{th1}.  Otherwise, we derive Problem {\bf C} by applying Lemma~\ref{lm2}, as
\begin{eqnarray} 
  \mbox{\bf (C) \bf maximize} && 
  \mbox{$\sum_{n \in \mathcal{U}}$} \log(1 + \gamma_n(t)) 
         \label{eq:optimalc} \\
  \mbox{\bf subject to:} && \nonumber \\
  && \gamma_n(t) = \frac{L_n P_n(t)}{\bar{P} - P_n(t) + A_{n}}, \; \mbox{for all} \; n 
         \label{eq:vbrctc1} \\
    && 
    P_{n}^{min}(t) \leq P_n(t) \leq P_{n}^{max}(t), \mbox{for all} \; n
        \label{eq:vbrctc2} \\
    && 
    \mbox{$\sum_{n \in \mathcal{U}}$} {P_n(t)} = \bar{P}, \label{eq:vbrctc3} 
\end{eqnarray}
where $A_n = \eta_n/G_n$ is the ratio of noise power and channel gain, representing the quality of the user $n$ downlink channel. $P_{n}^{min}(t)$ and $P_{n}^{max}(t)$ are solved from (\ref{eq:vbrct2}) and (\ref{eq:vbrctc1}), as 
\begin{equation} \label{eq:pmaxmin}
  \left\{ \begin{array}{l}
    P_{n}^{min}(t) = \gamma_n^{min} (\bar{P} + A_n) / (L_n + \gamma_n^{min}) \\
    P_{n}^{max}(t) = \gamma_n^{max} (\bar{P} + A_n) / (L_n + \gamma_n^{max}).
          \end{array} \right.
\end{equation} 

Since the total transmit power is $\bar{P}$, the objective value in (\ref{eq:optimalc}) and the SINR in (\ref{eq:vbrctc1}) for each user only depends on its own power. 
Note that all the constraints are now linear. To solve Problem {\bf C}, we examine the objective function to see if it is convex. We omit time index $t$ in the following for brevity. 

\begin{lemma} \label{lm3}
The capacity of each user $n$, $C_n$, has one {\em inflection point} $P_n^*$: when $P_n < P_n^*$, $C_n$ is in concave; when $P_n > P_n^*$, $C_n$ is convex. 
\end{lemma}
\begin{proof}
Taking the first and second derivatives of the objective function (\ref{eq:optimalc}) with respect to $P_n$, we have
\begin{eqnarray} \label{eq:derivatives}
  \frac{\partial{C_n(P_n)}}{\partial{P_n}}  &=& \frac{L_n(\bar{P}+A_n)}{(\bar{P}-P_n+A_n)[\bar{P}+(L_n-1)P_n+A_n]} \\
  \frac{\partial^2{C_n(P_n)}}{\partial{P_n}^2} &=& \frac{-L_n [(L_n-2)(\bar{P}+A_n)+2(1-L_n)P_n](\bar{P}+A_n)}{[(\bar{P}-P_n+A_n)^2+L_n P_n(\bar{P}-P_n+A_n)]^2}. 
\end{eqnarray}
Since $P_n \leq \bar{P}$ and $A_n > 0$, both the first and second derivatives exist. Letting $\frac{\partial^2{C_n(P_n)}}{\partial{P_n}^2} = 0$, we derive the unique inflection point 
\begin{equation} \label{eq:pnstar}
  P_n^* = \frac{L_n-2}{2(L_n-1)}(\bar{P}+A_n).
\end{equation}
When $P_n < P_n^*$, it can be shown that $\frac{\partial^2{C_n(P_n)}}{\partial{P_n}^2} < 0$; 
when $P_n > P_n^*$, it can be shown that $\frac{\partial^2{C_n(P_n)}}{\partial{P_n}^2} > 0$. 
\end{proof}

	\begin{figure} [t] 
	\center{\includegraphics[width=3.6in, height=2.6in]{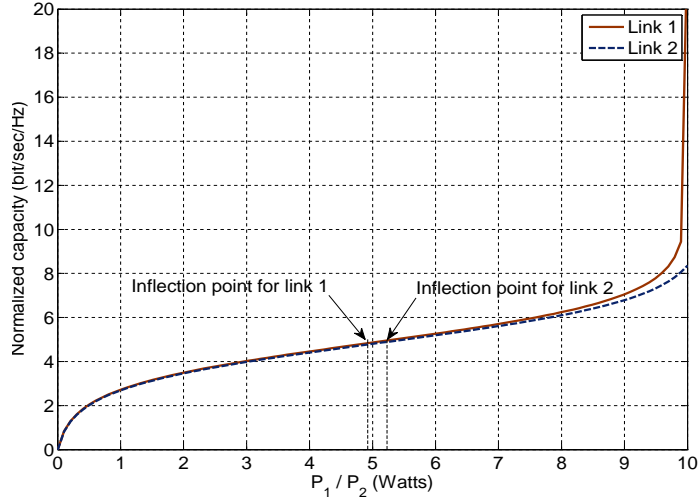}}
	\caption{Normalized capacity curves and inflection points for a two-user system, where link $1$ has better quality than link $2$, i.e. $A_1 < A_2$.}
	\label{fig:cp}
	\end{figure}

The normalized capacities for a two-user system is plotted in Fig.~\ref{fig:cp}, with the inflection points marked.  It can been observed that the curves are concave on the left hand side of the inflection points and convex on the right hand side of the inflection points. The processing gain is usually large for practical systems (e.g., $L_n = 128$ in IS-95 CDMA). 
We assume $L_n \gg 1$ in the following analysis.   

%
\begin{theorem} \label{th2}
For Problem {\bf C}, there can be at most two links operating in the convex region if $L_n \geq (4 \bar{P} + 6 A_n)/(\bar{P} + 3 A_n)$.  
\end{theorem}
\begin{proof}
The reflection point is $P_n^* = \frac{L_n-2}{2(L_n-1)}(\bar{P}+A_n)$. As $L_n \rightarrow \infty$, we have $P_n^* = 0.5(\bar{P} + A_n)$.  Only one link can operate in the convex region due to constraint (\ref{eq:vbrctc3}).  Since $\frac{\partial P_n^*}{\partial L_n} > 0$, $P_n^*$ is an increasing function of $L_n$.  When $1 \ll L_n < \infty$, we have $P_n^* < 0.5 \cdot (\bar{P} + A_n)$.  Letting $3 P_n^* = \bar{P}$, we have $L_n = (4 \bar{P} + 6 A_n)/(\bar{P} + 3A_n)$. 
\end{proof}


For a clean channel where $A_n \approx 0$, $L_n \geq 4$ will guarantee at most two links operating in the convex region.  The following results are on the impact of channel quality $A_n = \eta_n/G_n$.  

\begin{theorem} \label{th3}
For a given $L_n$, 
the inflection point $P_n^*$ is an increasing function of 
$A_n$. 
For two links $i$ and $j$ with the same transmit power $P$, if $A_i < A_j$, 
we have $C_i(P, A_i) > C_j(P, A_j)$ and $\frac{\partial{C_i(P_i, A_i)}}{\partial{P_i}} |_{P_i = P} > \frac{\partial{C_j(P_j,A_j)}}{\partial{P_j}}|_{P_j = P} > 0$. 
\end{theorem}
\begin{proof}
The first part can be easily shown by the first derivative of $P_n^*$ with respect to $A_n$, which is $\frac{\partial{P_n^*}}{\partial A_n}=\frac{L_n-2}{2(L_n-1)} > 0$, for $L_n>2$. The second part can be easily shown by evaluating (\ref{eq:optimalc}), (\ref{eq:vbrctc1}), and (\ref{eq:derivatives}). 
\end{proof}


Theorem~\ref{th3} shows that, for two links in the convex region with the same initial power $P$, allocating more power to the link with better quality can achieve larger objective value than alternative ways of splitting the power between the two links (i.e., achieving the multi-user diversity gain). Based on the above analysis, we develop the {\em second step} of the power allocation algorithm for solving Problem {\bf C}, as given in Table~\ref{tab:heuristic2}. 
In Table~\ref{tab:heuristic2}, Lines $3 \sim 4$ tests the feasibility of the power allocation. If the sum of the total minimum required power is larger than the BS peak power, there is no feasible power allocation 
and there will be buffer underflow.  
In this case, we select users with ``good'' channels for transmission and suspend the users with ``bad'' channels. 

\begin{table} [t] 
\begin{center}
\caption{Two-Step Power Allocation Algorithm: Step II}
\label{tab:heuristic2}
\begin{tabular}{ll}
	\hline
	   & \underline{Initialization}: \\
	1  & BS obtains $b_n$, $D_n$, and $B_n$ for all user $n$; \\
  2  & BS computes $\gamma_n^{max}$, $\gamma_n^{min}$, and $P_n^*$, for all $n$; \\
	3  & BS computes the minimum required sum power $\bar{P}_{min} = \sum_{n \in \mathcal{U}}{P_n^{min}}$ \\
	   & and gap $\Delta_P = \bar{P} - \bar{P}_{min}$; \\
	4  & IF $\bar{P}_{min} > \bar{P}$, remove links from $\mathcal{U}$, according to descending order \\
	   & of $A_n$, until $\bar{P}_{min} \leq \bar{P}$; \\
	5  & Compute $R_n = \frac{C_n(\min\{P_n^{max}, P_n^{min}+\Delta_P\}) - C_n(P_n^{min})}{\min\{P_n^{max}, P_n^{min}+\Delta_P\}-P_n^{min}},$ 
	 		 for all $P_n^{max} > P_n^*$;\\
	   & \underline{Phase 1}: \\
	6  & Select all the users satisfying $P_n^{min} < P_n^{*}$ as a set $\mathcal{U}'\subseteq \mathcal{U}$;\\
	7  & Solve Problem \textbf{C} under constraints \\
	   & $P_n^{min} \leq P_n \leq \min{(P_n^{max}, P_n^*)}$ and 
	   $\sum_{n \in \mathcal{U'}}{P_n} \leq \bar{P}' = \bar{P} - \sum_{n \in \bar{\mathcal{{U}}}'}{P_n^{min}}$,  \\
	   & where $\bar{\mathcal{U}}'$ is the complementary set of $\mathcal{U}'$, and obtain solution $\vec{P}_1$;\\
	8  & Calculate $R_n$ by updating $P_n^{min}$ to the solution in Line $7$ and assign \\
	   & the remaining power to the nodes in set $\mathcal{U}$, in descending order of $R_n$; \\
	9  & Obtain the Phase $1$ solution, $\vec{P}_{p_1}$, and objective value $f_{p_1}$; \\
	   & \underline{Phase 2}:\\
	10 & Select the link with the maximum $R_n$, and assign all the available \\
	   & power $\bar{P} - \bar{P}_{min}$ to the link, until either all the power is assigned or \\
	   & the link attains power $P_n^{max}$; \\
	11 & IF there is still power to allocate, THEN select all the nodes in set $\mathcal{U} \backslash{n}$ \\
	   & and repeat Lines $5 \sim 8$;\\
	12 & Obtain the Phase $2$ solution, $\vec{P}_{p_2}$, and objective value $f_{p_2}$; \\
	   & \underline{Phase 3}: \\
	13 & Select the first 2 links with the largest $R_n$'s, and assign all the availble \\
	   & power $\bar{P} - \bar{P}_{min}$ to the links, until all the power is assigned or the links \\
	   & attains power $P_n^{max}$, and repeat Line $11$; \\
	14 & Obtain the Phase $3$ solution, $\vec{P}_{p_3}$, and objective value $f_{p_3}$; \\
		 & \underline{Decision}:\\
	15 & Choose the largest objective value among $f_{p_1}$, $f_{p_2}$ and $f_{p_3}$, and stop \\
	   & with the corresponding power assignment; \\ 
	\hline
\end{tabular}
\end{center}
\end{table}


The Step II algorithm checks the three possible solution scenarios for Problem {\bf C} depending on the network status and video parameters: 
\begin{itemize}
  \item All links operate in the convex region; 
  \item One link operates in the convex region and the remaining links operate in the concave region 
  \item Two links operate in the convex region and the remaining links operate in the concave region. 
\end{itemize}
Each of the three phases in Table~\ref{tab:heuristic2} considers the optimality condition for one of the three scenarios.  
In particular, Phase 1 first optimizes the power allocation in the concave region and then allocates the remaining power to the links that could be moved to the convex region. 
Phase 2 allocates as much power as possible to the link with the best quality, which could work in the convex region. 
Phase 3 attempts to move the second best link to the convex region if the total power constraint is not violated.  
Usually when $L_n$ and $n$ are large, Phase 3 will rarely occur due to the peak power constraint.  

In Table~\ref{tab:heuristic2}, Line 7 presents a convex optimization component, for which several effective solution techniques can be applied. In the following section, we describe a distributed algorithm for Line 7 based on dual decomposition. 

\section{Distributed Algorithm} \label{sec:DistPower}

As discussed in Section~\ref{sec:PowerFrm}, the core of the Step II algorithm is to solve Problem {\bf C} in the concave region (see Fig.~\ref{fig:cp}).  In this section, we present a distributed algorithm for this purpose, where the users are involved in power allocation to reduce the control and computation overhead on the BS.  In the concave region, we have Problem {\bf D} as
\begin{eqnarray} 
  \hspace{-0.17in} \mbox{\bf (D) \bf maximize} && 
  \mbox{$\sum_{n \in \mathcal{U}}$} \log(1 + \gamma_n(t)) 
         \label{eq:optimald} \\
  \mbox{\bf subject to:} && \nonumber \\
  && \gamma_n(t) = \frac{L_n P_n(t)}{\bar{P} - P_n(t) + A_{n}}, \; \mbox{for all} \; n \label{eq:vbrctd1} \\
  && 
  P_{n}^{min}(t) \leq P_n(t) \leq \min\{P_n^{max}, P_n^*\}, \mbox{for all} \; n \label{eq:vbrctd2} \\
  && 
  \mbox{$\sum_{n \in \mathcal{U}}$} {P_n(t)} \leq P_{tot}, \label{eq:vbrctd3} 
\end{eqnarray}
where $P_{tot} \leq \bar{P}$ is the total power budget for the links 
in the concave region. For brevity, we define $P_n^{th} = \min\{P_n^{max}, P_n^*\}$ and drop the time slot index $t$ in the following analysis.  

Introducing non-negative Lagrange multipliers $\lambda_n, \mu_n$, and $\nu$ for constraints (\ref{eq:vbrctd2}) and (\ref{eq:vbrctd3}), respectively, we obtain the Lagrange function as
\begin{eqnarray}
&& \hspace{-0.1in} \mathcal{L}(\vec{P}, \vec{\lambda}, \vec{\mu}, \nu) \\
&& \hspace{-0.25in} = \mbox{$\sum_{n \in \mathcal{U}}$} \left[ \log \left( 1 + \frac{L_n P_n}{\bar{P} - P_n+A_n} \right) + \lambda_n(P_n - P_n^{min}) \right] + \nonumber\\
&& \hspace{-0.1in}  \mbox{$\sum_{n \in \mathcal{U}}$} \left[ \mu_n(P_n^{th} - P_n) \right] + \nu \left(P_{tot} - \mbox{$\sum_{n \in \mathcal{U}}$} P_n \right) \nonumber \\
&& \hspace{-0.25in} = \mbox{$\sum_{n \in \mathcal{U}}$} \left[ \mathcal{L}_n(P_n, \lambda_n, \mu_n, \nu) \hspace{-0.025in} + \hspace{-0.025in} (\mu_nP_n^{th} \hspace{-0.025in} - \hspace{-0.025in} \lambda_nP_n^{min}) \right] \hspace{-0.025in} + \hspace{-0.025in} \nu{P_{tot}}, \nonumber
\end{eqnarray}
where 
\begin{equation} \label{eq:ln}
  \mathcal{L}_n(P_n, \lambda_n, \mu_n, \nu) = \log \left( 1 + \frac{L_n P_n}{\bar{P} - P_n+A_n} \right) + (\lambda_n-\mu_n-\nu)P_n.
\end{equation}
Since $\mathcal{L}_n$ only depends on user $n$'s own parameters, we have the dual decomposition for each user $n$. For given Lagrange multipliers (or, prices) $\hat{\lambda_n}$, $\hat{\mu_n}$, and $\hat{\nu}$, we have the following subproblem for each user $n$.
\begin{eqnarray}
\hat{P}_n(\hat{\lambda_n}, \hat{\mu_n}, \hat{\nu}) = 
\stackbin[P_n^{min} \leq P_n \leq P_n^{th}]{}{\arg\max}
\mathcal{L}_n(P_n, \hat{\lambda_n}, \hat{\mu_n}, \hat{\nu}), \; \mbox{for all} \; n.  
\label{eq:sub}
\end{eqnarray}  
Subproblem (\ref{eq:sub}) has a unique optimal solution due to the strict concavity of $\mathcal{L}_n$. We use the gradient method~\cite{Bertsekas95} to solve (\ref{eq:sub}), where user $n$ iteratively updates its power $P_n$ as:
\begin{eqnarray} \label{eq:gradientP} 
&& \hspace{-0.00in} P_n(l+1) \\ 
&& \hspace{-0.15in} = \left[ P_n(l) + \theta(l) \nabla_n \mathcal{L}_n(P_n) \right]^* \nonumber \\
&& \hspace{-0.175in} = \left[ P_n(l) + \theta(l) \frac{L_n(\bar{P}+A_n)}{(\bar{P}-P_n+A_n)(\bar{P}+(L_n-1)P_n+A_n)} +  
\theta(l) (\lambda_n-\mu_n-\nu) \right]^*, \nonumber  
\end{eqnarray}
where $[\cdot]^*$ denotes the projection onto the range of 
$[P_n^{min}, P_n^{th}]$.  The update stepsize $\theta(l)$ varies in each step $l$ and is determine by the Armijo Rule~\cite{Bertsekas95}. Due to the strict concavity of $\mathcal{L}_n$, the series $\{P_n(1), P_n(2), \cdots \}$ will converge to the optimal solution $\hat{P}_n$ as $l \rightarrow \infty$.

For a given optimal solution for problem (\ref{eq:sub}), $\vec{\hat{P}}=[\hat{P}_1, \cdots, \hat{P}_N]^T$, the master dual problem is as follows:
\begin{eqnarray} 
  \mbox{ \bf minimize} && 
  \mathcal{L}(\vec{\hat{P}}, \vec{\lambda}, \vec{\mu}, \nu) 
         \label{eq:master} \\ 
  \mbox{\bf subject to:} && 
  \lambda_n, \mu_n, \nu \geq 0, \; \mbox{for all} \; n.
\end{eqnarray}
Since the objective function (\ref{eq:master}) is differentiable, we also apply the gradient method to solve the master dual problem~\cite{Bertsekas95}, where the Lagrange multipliers are iteratively updated as
\begin{equation} \label{eq:gradientL}
\left\{ \begin{array}{l}
  \lambda_n(l+1) = [\lambda_n(l) - \alpha_{\lambda}(l) \cdot \frac{\partial{\mathcal{L}(\vec{\lambda}, \vec{\mu}, \nu)}}{\partial{\lambda_n}}]^+, \; \mbox{for all} \; n \\
  \mu_n(l+1) = [\mu_n(l) - \alpha_{\mu}(l) \cdot \frac{\partial{\mathcal{L}(\vec{\lambda}, \vec{\mu}, \nu)}}{\partial{\mu_n}}]^+, \; \mbox{for all} \; n \\
  \nu(l+1) = [\nu(l) - \alpha_{\nu}(l) \cdot \frac{\partial{\mathcal{L}(\vec{\lambda}, \vec{\mu}, \nu)}}{\partial{\nu}}]^+, 
        \end{array} \right. 
\end{equation}
where $[\cdot]^+$ denotes the projection onto the nonnegative axis. The update stepsizes are also determined by the Armijo Rule~\cite{Bertsekas95}. As the dual variables $\vec{\lambda}(l), \vec{\mu}(l), \nu(l)$ converge to their stable values as $l \rightarrow \infty$, the primal variables $\vec{\hat{P}}$ will also converge to the optimal solution~\cite{Palomar06}.

The distributed algorithm is given in Table~\ref{tab:distalg}, where the above procedures are repeated iteratively.  The BS first broadcasts Lagrange multipliers to the users.  Each user updates its requested power as in (\ref{eq:gradientP}), using local information $P_n^{min}$, $P_n^{max}$, $P_n^*$, $A_n$, $L_n$, 
and BS peak power $\bar{P}$.  Each user then sends its requested power back to the BS, and the BS will updates the Lagrange multipliers as in (\ref{eq:gradientL}). And so forth, until the optimal solution is obtained. 

\begin{table} [t] 
\begin{center}
\caption{Distributed Power Control Algorithm}
\label{tab:distalg}
\begin{tabular}{ll}
	\hline
	1  & BS sets $l=0$ and prices $\lambda_n(l), \mu_n(l), \nu(l)$ equal to some \\
	   & nonnegative initial values for all $n$;\\
  2  & BS broadcasts the prices to the selected users;\\
	3  & Each user locally solves problem (\ref{eq:sub}) 
	     as in (\ref{eq:gradientP}) to obtain its \\
	   & requested power; \\
	4  & Each user sends its requested power to the BS; \\
	5  & BS updates prices $\lambda_n(l), \mu_n(l), \nu(l)$ 
	     as in (\ref{eq:gradientL}) and broadcasts \\
	   & new prices $\lambda_n(l+1), \mu_n(l+1), \nu(l+1)$ for all $n$; \\
	6  & Set $l = l+1$ and go to Step $3$, until the solution converges; \\

	\hline
\end{tabular}
\end{center}
\end{table}


\section{Simulation Results} \label{sec:simulation}

We evaluate the proposed algorithms with simulations, using 
a cellular network with 20 users. 
The downlink bandwidth is $1$ MHz. The path gain averages are $G_{n} = d_{n}^{-4}$, where $d_{n}$ is the physical distance from the BS to user $n$.  The downlink channel is modeled as log-normal fading with zero mean and variance $8$ dB~\cite{Lee05}.
The processing gains are set to $L_n =  128$ for all $n$. The distance
$d_{n}$ is uniformly distributed in [100m, 1000m]. The device temperature is $T_0 = 290$ Kelvin and the equivalent noise bandwidth is $B_w = 1MHz$. The BS peak power constraints is set to $\bar{P} = 10$ Watts.  We use three VBR movies traces, \textit{Star Wars}, \textit{NBC News}, and \textit{Tokyo Olympics}, 
from the Video Trace Lib~\cite{Videtrace}. 
Each playout buffer is set to $1.5$ times of the largest frame size in the requested VBR video.


In the simulations, the proposed power allocation algorithm is executed at the beginning of each time slot. In Fig.~\ref{fig:distpower}, we plot the cumulative consumption, overflow and transmission curves for \textsl{NBC News} transmitted to user 2. The top sub-figure is the overview of $10,000$ frames.  We also plot the curves from frame $2,620$ to $2,640$ in the bottom sub-figure. We observe that the cumulative transmission curve $X(t)$ is very close to the cumulative overflow curve $B(t)$, indicating that the algorithm always aim to maximize the transmission rate as allowed by the buffer and power constraints. The playout buffers are almost fully utilized most of the time. There is no playout buffer overflow and underflow for the entire range of $10,000$ frames.  Among the {\em NBC News} frames, frame $2,625$ is the largest frame. We let seven out of the $20$ links playout this largest frame simultaneously at time slot $2,625$ in the simulation. There is no buffer 
underflow under such heavy load. 

\begin{figure} [t] 
\center{\includegraphics[width=3.6in]{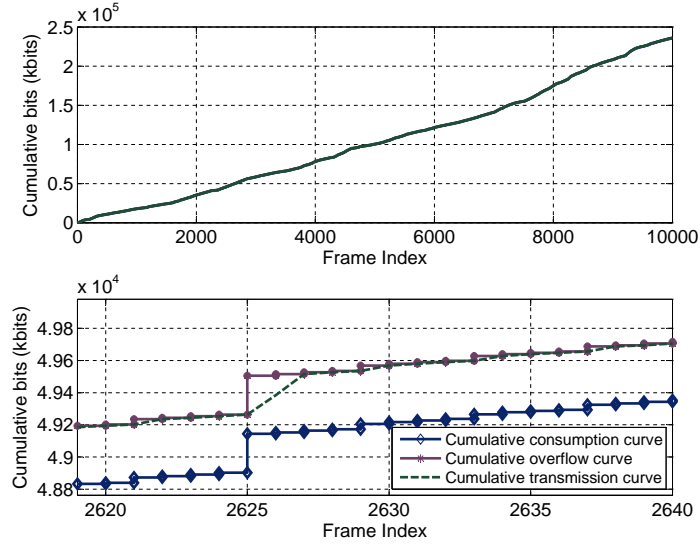}}
\caption{Transmission schedule for video \textsl{NBC News} to user 2.} \label{fig:distpower}
\end{figure}

In Fig. \ref{fig:distconverge}, we plot the power allocation and price updates for all the $20$ links in one of the 10,000 time slots. The power and prices converges in around $70$ steps. The converged power vector is 
$\vec{\hat{P}}$ 
= [0.0022, 1.396, 0.0356, 0.0024, 1.396, 0.0351, 0.0016, 1.396, 0.0356, 0.0026, 1.396, 0.0356, 0.0023, 1.396, 0.0356, 0.0018, 1.396, 0.0356, 0.0034, 1.394] Watts.

\begin{figure} [t] 
\center{\includegraphics[width=3.6in]{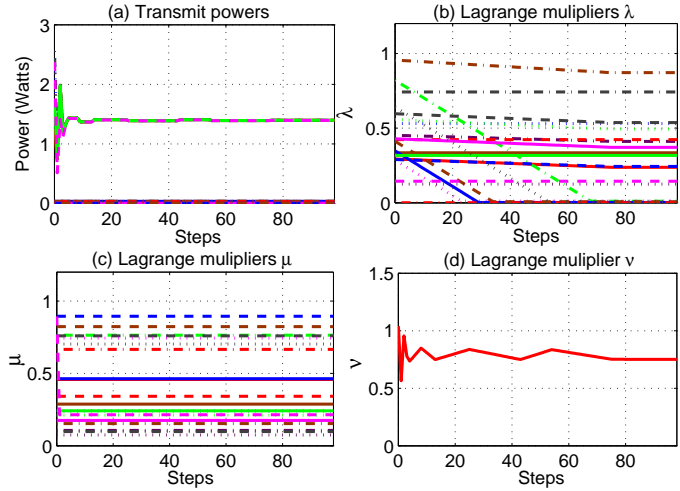}}
\caption{Convergence of power allocation and Lagrange multipliers.} \label{fig:distconverge}
\end{figure}

Finally, we compare the proposed algorithm with a diversity-aware power allocation scheme, where the BS allocates power according to channel quality. With this scheme, the best channel $n$ will be assigned power to achieve its maximum required power $P_{n}^{max}(t)$. Then the second best channel will be allocated power until its maximum required power is achieved, and so forth until all of $\bar{P}$ is allocated. 

We simulate $50$ users with the same network and video settings. We compare the algorithms by their average playout buffer utilization. In Fig.~\ref{fig:bufferutil}, we plot the average buffer utilization from frame $2,000$ to $2,999$.  It can be seen that the proposed algorithm consistently achieves high buffer utilization, ranging from 60\% to 100\%.  The diversity scheme achieves buffer utilization around 30\% except for frames from 2,250 to 2,400.  
Such considerably higher buffer utilization translates to better video quality: there is no buffer overflow or underflow for proposed algorithm, while there is buffer underflow in $17\%$ of the playout frames for the diversity scheme.
 
\begin{figure} [t] 
\center{\includegraphics[width=3.6in, height=2.6in]{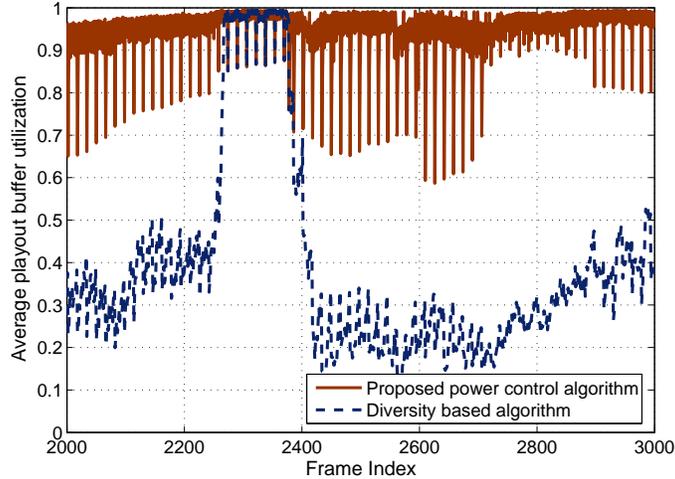}}
\caption{Average playout buffer utilization.} 
\label{fig:bufferutil}
\end{figure}

\section{Related Work} \label{sec:related}

There have been several papers on VBR video over wired network.
Due to long-range-dependent (LRD) VBR video traffic, the piecewise-constant-rate transmission and transport (PCRTT) method was used to optimize 
certain objectives while preserving continuous video playout. 
In~\cite{Liew97}, Liew and Chan developed bandwidth allocation schemes for multiple VBR videos to share a CBR channel. 
In~\cite{Salehi98}, Salehi et al. presented an optimal algorithm 
for smoothing VBR video over a CBR link. 
Feng and Liu~\cite{Feng00} introduced a critical bandwidth allocation algorithm to reduce the number of bandwidth variations and to maximize receiver buffer utilization.  
Due to the fundamental difference between wireless and wired links, these techniques cannot be directly applied to the problem of VBR video over wireless networks.

The downlink power allocation problem was studied in~\cite{Lee05, Lee09}, aiming to obtain the power allocation that maximizes a properly defined system utility. A distributed algorithm based on dynamic pricing and partial cooperation was proposed. 
Deng, Webera, and Ahrens~\cite{Deng09} studied the achievable maximum sum rate of multi-user interference channels. 
These papers provide the theoretical foundation and effective algorithms for utility maximization of downlink traffic, but the techniques used 
cannot be directly applied for VBR video over wireless networks with buffer and delay constraints.

In~\cite{Stockhammer04, Liang07}, the authors studied the problem of one VBR stream over a given time-varying wireless channel. In~\cite{Stockhammer04}, it was shown that the separation between a delay jitter buffer and a decoder buffer is in general suboptimal, and several critical system parameters were derived. 
In~\cite{Liang07}, the authors studied the frequency of jitters under both network and video system constraint and provided a framework for quantifying the trade-offs among several system parameters.
%
In this paper, we jointly consider power control in wireless networks, playout buffers, and video frame information, 
and address the more challenging problem of streaming multiple VBR videos, and present a cross-layer optimization approach that does not depend on any specific channel or video traffic models.  

\section{Conclusion} \label{sec:conclusion}

We developed a downlink power allocation model for streaming multiple VBR videos in a cellular network. The model considers interactions among downlink power control, channel interference, playout buffers, and VBR video traffic characteristics.  
The formulated problem aims at maximizing the total transmission rate under both peak power and playout buffer overflow/underflow constraints.  
We presented a two-step approach for solving the problem and a distributed algorithm based on the dual decomposition technique. 
Our simulation studies 
validated the efficacy of the proposed algorithms.

\section*{Acknowledgment}

Shiwen Mao's research is supported in part by the National Science Foundation under Grants CNS-0953513, ECCS-0802113, IIP-1032002 and CNS-0855251, and through the Wireless Internet Center for Advanced Technology at Auburn University (under NSF Grant IIP-0738088).



\end{document}